\documentclass{sig-alternate-per}

\usepackage{blindtext}
\usepackage{array,xcolor,colortbl,graphicx,multirow}
\usepackage{comment}
\usepackage{balance}
\usepackage{algorithm}
\usepackage[noend]{algpseudocode}
\usepackage[font={footnotesize}]{subcaption}
\usepackage[font={footnotesize}]{caption}
\usepackage{breakcites}
\usepackage{booktabs}
\usepackage{xspace}

\newtheorem{observation}{Observation}

\newcommand{\ALG}{\textsc{Bma}\xspace}
\newcommand{\OPT}{\textsc{Opt}\xspace}
\newcommand{\OFF}{\textsc{Off}\xspace}
\newcommand{\DET}{\textsc{Det}\xspace}
\newcommand{\On}{\textsc{On}\xspace}

\newcommand{\FixSaturation}{\textsc{FixSaturation}\xspace}
\newcommand{\FixMatching}{\textsc{FixMatching}\xspace}
\newcommand{\cost}{\textsc{cost}}
\newcommand{\cnt}[1]{\ensuremath{h_{#1}}}
\newcommand{\req}[1]{\ensuremath{H_{#1}}}
\newcommand{\matched}[1]{\ensuremath{M_{#1}}}
\newcommand{\edges}[1]{\ensuremath{E_{#1}}}
\newcommand{\edgesp}[1]{\ensuremath{E'_{#1}}}
\newcommand{\thresh}[1]{\ensuremath{T_{#1}}}
\newcommand{\VV}{\ensuremath{V^2}}
\newcommand{\assign}[1]{\ensuremath{A_{#1}}}
\newcommand{\assigne}[2]{\ensuremath{A_{#1}(#2)}}
\newcommand{\I}{\ensuremath{\mathcal{I}}}

\newtheorem{lemma}{Lemma} 
\newtheorem{theorem}{Theorem}

\title{Online Dynamic B-Matching\thanks{Research supported by the European Research Council
(ERC) under grant agreement number 864228 (AdjustNet project)
and the Polish National Science Centre grants 2016/22/E/ST6/00499
and 2015/18/E/ST6/00456.}}
\subtitle{With Applications to Reconfigurable Datacenter Networks
}

\numberofauthors{2}

\author{
\alignauthor Marcin Bienkowski\\
       \affaddr{University of Wroc{\l}aw, Poland}\\
       \email{marcin.bienkowski@cs.uni.wroc.pl} 
\alignauthor David Fuchssteiner\\
       \affaddr{University of Vienna, Austria}\\
       \email{david.alexander.fuchssteiner@univie.ac.at}
\and
\alignauthor Jan Marcinkowski\\
       \affaddr{University of Wroc{\l}aw, Poland}\\
       \email{jan.marcinkowski@cs.uni.wroc.pl}
\alignauthor Stefan Schmid\\
       \affaddr{University of Vienna, Austria}\\
       \email{stefan\_schmid@univie.ac.at}
}

\begin{document}

\maketitle
\begin{abstract}
This paper initiates the study of online algorithms for the 
maximum weight $b$-matching problem, a generalization of maximum weight matching
where each node has at most $b \geq 1$ adjacent matching edges. 
The problem is motivated by emerging optical technologies which allow
to enhance datacenter networks with reconfigurable
matchings, providing direct connectivity between
frequently communicating racks. These additional links
may improve network performance, by leveraging  
spatial and temporal structure in the workload.
We show that the underlying algorithmic problem
features an intriguing connection to online paging (a.k.a.~caching),
but introduces a novel challenge. 
Our main contribution is an online algorithm which 
is $O(b)$-competitive; we also prove that 
this is asymptotically optimal. 
We complement our theoretical results with extensive trace-driven
simulations, based on real-world datacenter workloads as well as
synthetic traffic traces. 
\end{abstract}



\keywords{online algorithms, competitive analysis, b-matching, reconfigurable networks, demand-aware networks}

\sloppy


\section{Introduction}\label{sec:intro}

\subsection{Motivation: Reconfigurable Datacenters}

The popularity of distributed data-centric applications 
related to machine learning and AI has led to an explosive
growth of datacenter traffic, and researchers are hence
making great efforts to design more efficient datacenter networks,
providing a high throughput at low cost.
Indeed, over the last
years, much progress has been made in the design of innovative datacenter
interconnects, based on 
fat-tree topologies~\cite{alFares2008,f10}, hypercubes~\cite{bcube,mdcube},
expanders~\cite{xpander} or random graphs~\cite{jellyfish},
among many others~\cite{dc-survey-1,dc-survey-2}.
All these networks have in common that their topology
is static and fixed.

An emerging intriguing alternative to these static datacenter networks
are reconfigurable networks~\cite{ccr18san,firefly,helios,sigact19,projector,DBLP:conf/hotnets/KandulaPB09,reactor,opera,rotornet,eclipse,xweaver,zhou2012mirror}: networks whose topology can be 
changed \emph{dynamically}.
In particular, novel optical technologies
allow to provide ``short cuts'', i.e., 
direct connectivity between top-of-rack 
switches,
based on dynamic matchings. 
First empirical studies demonstrate the potential of 
such reconfigurable networks, which can 
deliver very high bandwidth efficiency at low cost.

The matchings provided by reconfigurable networks
are either periodic (e.g.,~\cite{ballani2020sirius,opera,rotornet})
or demand-aware (e.g.,~\cite{avin2020online,dan,avin2019renets,firefly,helios,ancs18,projector}).
The latter is attractive as it allows
to leverage structure in the demand: datacenter traffic
is known to be highly structured, e.g., traffic
matrices are typically sparse and some flows (sometimes called elephant flows)
much larger than others~\cite{benson2010understanding,roy2015inside}.
This may be exploited: in principle, demand-aware datacenter networks
allow to directly match racks which communicate more
frequently,  leveraging spatial and
temporal locality in the workload~\cite{sigmetrics20complexity}.
These reconfigurable matchings are usually assumed to 
enhance a given fixed datacenter topology
based on traditional electric switches~\cite{projector}: 
the remaining traffic (e.g., mice flows) can be routed along the
fixed network, e.g., using classic shortest path control planes such as ECMP.

The advent of such hybrid static-dynamic datacenter networks introduces an
online optimization problem: how to enhance a given fixed topology with a set of
additional shortcut ``demand-aware'' edges, such that the current demand is served optimally
(e.g., large flows are routed along short paths, minimizing
the ``bandwidth tax''~\cite{rotornet}), while at the same time reconfiguration costs are kept
minimal  (reconfigurations take time and can 
temporarily lead to throughput loss).

\subsection{Problem in a Nutshell}

The above problem can be modeled as an online dynamic version of the classic
$b$-matching problem~\cite{anstee1987polynomial} (where $b$ is the number of
optical switches). In this problem, each node can be connected with at 
most $b$ other nodes (using optical links), which results in a $b$-matching.

Interestingly, while the offline version of the $b$-matching problem has been
studied intensively in the past (e.g., in the context of
matching applicants to posts)~\cite{Schrij03}, we are not aware of any work on the dynamic
online variant. As the problem is fundamental and finds applications beyond the
reconfigurable datacenter design problem, we present it
in the following abstract form.

\smallskip
\noindent\textbf{Input.} We are given an arbitrary (undirected) static weighted
and connected network on the set of nodes $V$ connected by a set of
non-configurable links $F$: the fixed network. Let $\VV$ be the set of
all possible unordered pairs of nodes from~$V$. For any pair $\tau = \{u,v\} \in
\VV$, let $\ell_\tau$ denote the length of a
shortest path between nodes $u$ and~$v$ in graph $G=(V,F)$.
 Note that $u$ and $v$
are not necessarily directly connected in~$F$.

The fixed network can be enhanced with reconfigurable links, providing a matching of degree $b$: 
Any node pair from $\VV$ may become a \emph{matching edge} 
(such an edge corresponds 
to a reconfigurable optical link), but 
the number of matching edges adjacent to any node has to be at most $b$,
for a given integer $b \geq 1$.

The demand is modelled as a sequence of 
communication requests\footnote{A request could either be an individual packet or a certain amount of data transferred.
This model of a request sequence is often considered in the literature and is more fine-grained than, e.g., a sequence of traffic matrices.} 
$\sigma=\{s_1,t_1\},\{s_2,t_2\}, \ldots$ revealed over time,
where $\{s_i,t_i\} \in V^2$.

\smallskip
\noindent\textbf{Output and Objective.} The goal is to 
schedule the reconfigurable links over time, that is,
to maintain 
a dynamically
changing $b$-matching $M \subseteq \VV$. Each node pair from $M$ is called 
a \emph{matching edge} and we require that each node has
at most $b$ adjacent matching edges. We aim to jointly minimize routing
and reconfiguration costs, defined below.

\smallskip
\noindent\textbf{Costs.} 
The routing cost for a request $\tau = \{s,t\}$ depends on whether $s$ and $t$ are
connected by a matching edge. In this model, a~given request can
either only take the fixed network or a direct matching edge (i.e., routing is
segregated~\cite{projector}). If $\tau \notin M$, the requests is routed exclusively on the
fixed network, and the corresponding cost is~$\ell_\tau$
(shorter paths imply smaller resource costs, i.e., lower ``bandwidth tax''~\cite{rotornet}). 
If $\tau \in M$, the
request is served by the matching edge, and the routing costs 0 (note that this is the most challenging cost function:
our result only improves if this cost is larger). 

Once the request is served, an algorithm may modify the set of matching edges: 
reconfiguration costs $\alpha$ per each node pair added or removed from the
matching $M$. (The reconfiguration cost and time can be assumed to be independent
of the specific edge.)

\smallskip
\noindent\textbf{Online algorithms.} 
An algorithm $\On$ is \emph{online}
if it has to take decisions without knowing the future requests
(in our case, e.g., which edge to include next in the matching
and which to evict). 
Such an algorithm is said to be $\rho$-competitive~\cite{BorEl-98} 
if there exists a~constant $\beta$ such that 
for any input instance~$\I$, it holds that 
\[
	\cost(\On,\I) \leq \rho \cdot \cost(\OPT,\I) + \beta \,,
\]
where $\cost(\OPT,\I)$ is the cost of the optimal (offline) solution 
for~$\I$. It is worth noting that $\beta$ can depend on 
the parameters of the network, such as the number of nodes, 
but has to be independent of the actual sequence of requests. 
Hence, in the long run, this additive term $\beta$ becomes negligible
in comparison to the actual cost of online algorithm $\On$.

\newpage

\subsection{Our Contributions}

This paper initiates the study of a
natural problem, online dynamic $b$-matching.
For example, this problem finds applications 
in the context of emerging reconfigurable
datacenter networks.

We make the following contributions:
\begin{itemize}
	\item We show that the online dynamic $b$-matching problem
	features an interesting connection to online paging problems,
	however, with a twist, introducing a new challenge.
	\item We present an $O((1+\ell_{\max}/\alpha) \cdot b)$-competitive deterministic algorithm, where $\ell_{\max} = \max_{e \in \VV} \ell_e$. 
	Note that in all relevant practical applications $\ell_{\max} \ll \alpha$ 
	(i.e., the cost of routing between any two nodes is much smaller than the cost 
	of reconfiguration, and hence the term $ \ell_{\max}/\alpha$ is negligible).
	\item We derive a lower bound which
	shows that no deterministic algorithm can achieve a
	competitive ratio better than $b$.
	\item We verify our approach experimentally, performing
	extensive trace-driven simulations, based on real datacenter workloads
	as well as synthetic traffic traces.
\end{itemize}


\subsection{Challenges, Technical Novelty, Scope}

At the heart of our approach lies the observation that online $b$-matching is
similar to online paging~\cite{AcChNo00,FKLMSY91,McGSle91,SleTar85}: each node
in the network can manage its reconfigurable edges in a cache of size~$b$.
However, making a direct reduction to caching seems impossible as reconfigurable edges
involve \emph{both} incident nodes, which introduces non-trivial dependencies.
Without accounting for these dependencies, the competitive ratio would be in the
order of the total number of reconfigurable edges in the network, whereas we in
this paper derive results which only depend on the number of per-node edges. Our
algorithms hence combine ``per-node caches'' in a clever way.

Generally, we believe that the notion of link caching 
has interesting implications for reconfigurable network designs
beyond the model considered in this paper. In particular, caching strategies
can typically be implemented locally, and hence may allow to overcome
centralized control overheads, similar to the stable matching algorithms
proposed in the literature~\cite{projector}. However, we leave the discussion of 
of such decentralized schedulers to future work. 

\subsection{Organization}

The remainder of this paper is organized as follows.
Our online algorithm is described and analyzed in 
Section~\ref{sec:algo},
and the lower bound is presented in Section~\ref{sec:lower}.
We report on our simulation results in 
Section~\ref{sec:simulations}.
After reviewing related literature in Section~\ref{sec:relwork},
we conclude our contribution in Section~\ref{sec:conclusion}.


\section{Algorithm BMA}
\label{sec:algo}

\algdef{SE}[SUBALG]{Indent}{EndIndent}{}{\algorithmicend\ }%
\algtext*{Indent}
\algtext*{EndIndent}
\begin{algorithm*}[t]
\caption{Algorithm \ALG}
\label{alg:main}
\begin{algorithmic}[1]
	\State{\textbf{Initialization:}} 
	\Comment{\emph{Matching is empty and counters are zero}}
	\Indent
	\State{$M \gets \emptyset$}
	\For{each edge $e$}
		
		\State{$\cnt{e} \gets 0$} 
	\EndFor
	\EndIndent

	\State{}

	\State{\textbf{Request $\tau = \{u,v\}$ arrives:}} 
	\Indent
	\If{$\tau \notin M$}
		\State{$\cnt{\tau} \gets \cnt{\tau} + 1$}		
		\If{$\cnt{\tau} = \thresh{\tau}$}
				\Comment{\emph{If $\tau$ becomes saturated,}}
			\State{Execute $\FixSaturation(u,\tau)$}
			\State{Execute $\FixSaturation(v,\tau)$}
			\If{$\cnt{\tau} = \thresh{\tau}$}
					\Comment{\emph{and if no desaturation event occured,}}
				\State{Execute $\FixMatching(u)$}
				\State{Execute $\FixMatching(v)$}
				\State{\label{line:add_to_matching} $M \gets M \cup \{ \tau \}$}
					\Comment{\emph{add $\tau$ to the matching.}}
			\EndIf
		\EndIf
	\EndIf
	\EndIndent
	
	\State{}
	\State{\textbf{Routine} $\FixSaturation(w,\tau)$\textbf{:}}
	\Indent
	\State{$\edgesp{w} = \edges{w} \setminus \{ \tau \}$}
	\If{\label{line:too_many_marked} 
			$|\edgesp{w} \cap \{ e : \cnt{e} = \thresh{e} \}| \geq b$}
				\Comment{\emph{If the number of saturated node pairs from $\edgesp{w}$ is at least $b$,}}
				\For{\label{line:unmarking} each edge $e \in \edges{w}$}
				\Comment{\emph{reset counters of all node pairs from $\edges{w}$ 
					(desaturation event at $w$).}}
			\State{\label{line:unmarking2} $\cnt{e} \gets 0$} 
		\EndFor
	\EndIf	
	\EndIndent
	
	\State{}
	\State{\textbf{Routine} $\FixMatching(w)$\textbf{:}}
	\Indent
	\If{\label{line:matching_too_large} $|M \cap \edges{w}| = b$}
		\Comment{\emph{If there are already $b$ incident matching edges,}}
		\State{\label{line:find_q_prime} 
			Pick any $e^* \in M \cap \edges{w}$ such that $\cnt{e^*} < \thresh{e^*}$}
			\Comment{\emph{remove any unsaturated edge $e^*$ from the matching.}}
		\State{$M \gets M \setminus \{e^*\}$}
	\EndIf
	\EndIndent
\end{algorithmic}
\end{algorithm*}

This section introduces our online $b$-matching 
algorithm, together with a competitive analysis.
As described above, in our case study of reconfigurable
networks, the matching links may for example describe the 
reconfigurable links provided by optical circuit switches, 
offering shortcuts between datacenters racks. 

Before we present the algorithm and its analysis
in details, let us first provide some intuition of our approach
and the underlying challenges. 
To this end, let us for now assume that the fixed network
$G=(V,F)$ is a complete unweighted graph (e.g., capturing the
distances in a datacenter network),
that $\alpha = 1$,
and that all requests are node pairs involving one chosen
node $w$. 
Also recall that each node can have at most~$b$ incident matching edges.

In this simplified scenario, we can observe that 
the choice of an appropriate set of matching edges becomes essentially a
variant of online caching (more precisely, a~variant of online paging with
bypassing~\cite{EpImLN15}). That is, an~algorithm maintains a set of at most $b$ edges 
incident to $w$ that are
in the matching. These edges can be thought as ``cached'': 
subsequent requests
to matched (cached) edges do not incur further cost.
Thus, from the perspective of a single node, the question is roughly equivalent
to maintaining a cache of at most $b$ items
(which leads to the typical algorithmic questions such as 
which item to cache or evict next).

However, if we simply run independent paging algorithms 
at all nodes, local
perspectives of particular nodes might not be coherent with each other: one
endpoint of a node may want to keep an~edge in the matching, while the other may
want to evict it from the matching. 
This is practically undesirable, as transmitters and receivers typically need
to be aligned and coordinated~\cite{projector,rotornet}.
To illustrate this issue, assume that node
$w$ wants to add a~new matching edge $\{w,w'\}$, but it already has $b$ incident
matched edges. To accommodate a new matching edge, $w$ removes edge $\{w,w''\}$
from the matching. This however removes the edge not only from the ``cache'' of
node $w$, but also from the ``cache'' of node~$w''$. Handling this coherence
issue with a low overall cost, constitutes a main technical 
challenge that we need to tackle in
our algorithm.

\subsection{Algorithm Definition}

Our algorithm \ALG is defined as follows. For each node pair~$e \in \VV$, \ALG keeps
a counter $\cnt{e}$, initially equal to zero. The value of $\cnt{e}$ will always
be a lower bound for the number of times $e$ has been requested since it was removed from the matching the last time
(or from the beginning of the input sequence
if $e$ was never in the matching). For each node pair $e \in \VV$ we define a
threshold 
\[ \thresh{e} = 2 \cdot \lceil \alpha/\ell_e \rceil . \] 
Once the counter
$\cnt{e}$ reaches $\thresh{e}$, and additional certain conditions are fulfilled,
edge $e$ will be added to the matching. 
Otherwise, the counter will be reset to zero. A node pair $e \in \VV$ whose
counter value is equal to $\thresh{e}$ is called \emph{saturated}; our algorithm
always keeps all saturated node pairs in the matching. Note that these 
$\thresh{e}$ requests to node pair $e$ induced a total cost of 
$\thresh{e} \cdot \ell_e \in [2 \alpha, 2 (\alpha / \ell_e + 1) \cdot \ell_e] = [2 \alpha, 2 \alpha + 2 \ell_e]$.

For any node~$w$, we define $\edges{w} = \{\{w,v\} : v \in V\}$, i.e., 
$\edges{w} \subseteq \VV$ is the set of all node pairs, with one node equal to $w$. 
Recall that, at any time, $\matched{} \subseteq \VV$ denotes the set of matching edges. 

\ALG is designed to preserve three
invariants:
\begin{description}
	\item[Counter invariant:] $0 \leq \cnt{e} \leq \thresh{e}$.
	\item[Saturation invariant:] If $\cnt{e} = \thresh{e}$, then $e \in M$.
	\item[Matching invariant:] If $e \in M$, then $\cnt{e} = \thresh{e}$
		or $\cnt{e} = 0$.
\end{description}
plus an invariant for any node $w$:
\begin{description}
	\item[Saturation degree invariant:] $|\{e \in \edges{w} : \cnt{e} = \thresh{e}\}| \leq b$.
\end{description}

A pseudo-code for our algorithm \ALG is given in Algorithm~\ref{alg:main}.
In the next section, we will explain it in more details,
and prove that it 
never violates the invariants.

\subsection{Maintaining Invariants}

At the beginning, the matching $M$ is empty and the counters of all edges are
zero, and thus all invariants hold. Below we describe what happens upon serving
a communication request $\tau$ by \ALG and how \ALG ensures that all invariants
are preserved.

\ALG first verifies whether $\tau$ is a matching edge. If so, then such a
request incurs no cost and \ALG does nothing. Otherwise, by the saturation
invariant, $\cnt{\tau} \leq \thresh{e} - 1$. In this case, $\ALG$ pays for the 
communication request and increments counter $\cnt{\tau}$. The increment preserves 
the counter invariant. The matching invariant holds emptily as $\tau \notin M$. 
If $\cnt{\tau}$ is still below $\thresh{e}$, then the remaining invariants also hold 
and \ALG does not execute any further actions.

However, if the value of $\cnt{\tau}$ reaches $\thresh{\tau}$ ($\tau$ becomes
saturated), the 
saturation invariant becomes violated ($\tau$ should be a matching
edge) and also
the saturation degree invariants may become violated at $u$ and $v$. We
now explain how \ALG handles these issues.

\ALG first ensures that the saturation degree invariant is satisfied at both
endpoints of $\tau$. To this end, \ALG executes
the \FixSaturation routine at
$u$ and $v$. If at any node $w \in \{u,v\}$ 
the number of saturated node pairs from $\edges{w}$ 
different from $\tau$
is already $b$, all edges of~$\edges{w}$, including~$\tau$, have
their counters reset to zero. In this 
case, we say that a \emph{desaturation
event} occurred at the respective endpoint~$w$. Note that all four cases are possible:
there can be no desaturation event, a~desaturation event can occur at $u$, at
$v$, or at both of $u$ and $v$. 
The execution of the \FixSaturation routines reestablishes the
saturation degree invariants, and preserves counter, saturation and matching
invariants at all edges different from $\tau$.

For edge $\tau$, the corresponding counter and matching invariants clearly hold.
If any desaturation event occurs, then it also fixes 
the saturation invariant for
$\tau$, and \ALG need not do anything more. 
Otherwise (no desaturation event
occurs, that is, $\cnt{\tau}$ is still equal to $\thresh{\tau}$), the saturation
invariant is violated: $\tau$ has to be added to the matching. However, if any
of $\tau$ endpoints already has $b$ incident matching edges, one such edge has
to be removed from the matching before. This is achieved
by the \FixMatching routines
executed at both $u$ and~$v$: if necessary, they remove one incident
non-saturated matching edge. It remains to show that such 
a matching edge indeed
exists. 

\begin{lemma}
\label{lem:alg_well_defined}
A non-saturated matching edge $e^*$ chosen at Line~\ref{line:find_q_prime} of
the algorithm \ALG (in routine $\FixMatching(w)$) always exists. Moreover, when
it is removed, $\cnt{e^*} = 0$.
\end{lemma}

\begin{proof}
Assume that $\FixMatching(w)$ is executed (for a node $w \in \{u,v\}$). 
Let $S_w = \edges{w} \cap \{ e : \cnt{e} = \thresh{e} \}$ 
be the set of saturated node pairs from $\edges{w}$.
Note that $\FixMatching(w)$ is preceded
by the execution of $\FixSaturation(w,\tau)$: it ensures that $|S_w| \leq b$
($S_w$~contains $\tau$ and at most $b-1$ other edges). 

Let $M_w = M \cap \edges{w}$ be the set of matching edges 
incident to $w$. The condition at Line~\ref{line:matching_too_large} ensures that
$|M_w| = b$.

Now, observe that the set
$S_w \setminus M_w$ is non-empty as it contains the requested node pair $\tau$.
However, as $|M_w| \geq |S_w|$, the set $M_w \setminus S_w$ is non-empty either,
and any of its edges is a viable candidate for~$e^*$.

Finally, by the matching invariant, the counter of a~matching edge $e \in M_w$ 
is equal either to $\thresh{e}$ or $0$. Hence, the counters of all matching edges in $M_w \setminus S_w$
are zero, and thus $\cnt{e^*} = 0$.
\end{proof}

\subsection{Desaturation Events}

Fix any node $w$. For the analysis of \ALG, a natural approach would be to
estimate the number of paid requests to all node pairs of $\edges{w}$ between
two desaturation events at $w$. This number corresponds to the total increase of all
counters corresponding to these node pairs in the considered time interval.
However, such an approach fails as these counters may be reset multiple times
because of desaturation events at other nodes. In particular, it is possible
that a node pair $\{w,u\}$ from $\edges{w}$ in included multiple times in the
matching between two desaturation events at $w$. Therefore, we develop a~more
complicated accounting scheme.

First, not only we track counters, but for any node pair~$e$ we keep track
of a set $\req{e}$ of requests paid by \ALG (that caused the increase 
of the counter $\cnt{e}$, i.e., $|\req{e}| = \cnt{e}$). 
When the counter $\cnt{e}$ is reset, the set $\req{e}$ is emptied.

When requests paid by \ALG become removed from sets $\req{e}$, we map them to 
the corresponding desaturation 
events: for any desaturation event $d$, we create a~set of requests
$\assign{d}$, so that all these sets are disjoint. Requests that still belong to the
current contents of sets $\req{e}$ are not (yet) mapped. More precisely, note
that when a~desaturation event at a node $w$ occurs, we empty all sets $\req{e}$
for node pairs $e \in \edges{w}$. If a request $\tau = \{u,v\}$ triggers a
single desaturation event $d_u$ at $u$, then we simply set $\assign{d_u} =
\biguplus_{e \in \edges{u}} \req{e}$, i.e., we map all requests corresponding
to counters that were reset by~$d_u$. If, however, a request $\tau = \{u,v\}$
triggers desaturation events $d_u$ and $d_v$ both at $u$ and $v$, we want
requests from $\req{\tau}$ to be mapped (partially) to both desaturation
events. Thus, we partition $\thresh{\tau}$ requests from $\req{\tau}$
arbitrarily into two subsets $\req{\tau}^v$ and $\req{\tau}^u$, each of
cardinality $\thresh{\tau}/2 = \lceil \alpha / \ell_e \rceil$, and set
$\assign{d_u} = \req{\tau}^u \cup \biguplus_{e \in \edges{u} \setminus \tau}
\req{e}$ and $\assign{d_v} = \req{\tau}^v \cup \biguplus_{e \in \edges{v}
\setminus \tau} \req{e}$.

For any request set $P$, let $\ell(P) = \sum_{e \in P} \ell_e$, i.e.,
$\ell(P)$ is the cost of serving all requests from $P$ without using matching edges.
For any desaturation event $d$ and a node pair $e$, let $\assigne{d}{e}$ be
the requests of~$\assign{d}$ to node pair $e$. 
The following observation follows immediately by the definition of \ALG
and sets $\assign{d}$. 

\begin{observation}
\label{obs:desaturation}
Fix any desaturation event $d$ at any node $w$. Then, the following
properties hold:
\begin{enumerate}
\item \label{item:at_most_tresh}
	For any node pair $e \in \edges{w}$, it holds that $|\assigne{d}{e}| \leq 
	\thresh{e}$. 
\item \label{item:b_or_b1_with_thresh}
	There exists a set $P \subseteq \edges{w}$ of cardinality $b+1$, such 
	that $|\assigne{d}{e}| \geq \thresh{e}/2$ for each $e \in P$.
\end{enumerate}
\end{observation}

\subsection{Competitive Ratio of BMA}

We now use sets $\assign{d}$ to estimate the costs of \ALG and \OPT. 
We do not aim at optimizing the constants, but rather at the
simplicity of the argument.

\begin{lemma}
\label{lem:alg_cost}
Let $D(\I)$ be the set of all desaturation events that occurred during 
input $\I$.
Then, 
\[ 
	\cost(\ALG,\I) \leq 4 \cdot |\VV| \cdot (\alpha + \ell_{\max}) + 
	2 \sum_{d \in D(\I)} \ell(\assign{d}) \,.
\]
\end{lemma}

\begin{proof}
Within this proof, we consider contents of sets $\req{e}$
right after \ALG processes the whole input $\I$.
For any node pair $e \in \VV$, the set $\req{e}$ contains at most $\thresh{e}$
edges, and therefore
\[
	\ell(\req{e}) \leq \thresh{e} \cdot \ell_{e} \leq 2 \cdot \alpha + 2 \cdot \ell_e 
	\leq 2 \cdot (\alpha + \ell_{\max}) \,.
\]

Any request to a node pair~$e$ paid by \ALG is either in set $\req{e}$ 
or it was already assigned to a set $\assign{d}$ for some desaturation event $d \in D(\I)$.
Thus, the cost of serving all requests by \ALG is at most 
\[
	\sum_{e \in \VV} \ell(\req{e}) + \sum_{d \in D(\I)} \ell(\assign{d}) 
		\leq 2 \cdot |\VV| \cdot (\alpha + \ell_{\max}) + \sum_{d \in D(\I)} 
		\ell(\assign{d}).
\]

To bound the cost of matching changes, 
we observe that by the definition of \ALG, only a saturated node pair $e$
may become included in the matching. If $e$ becomes removed from the matching later,
then by Lemma~\ref{lem:alg_well_defined}, the counter of $e$ must have dropped 
to zero in the meantime. Therefore, any addition of $e$ to the matching can 
be mapped to the unique $\thresh{e}$ paid requests to $e$.
As the cost of such $\thresh{e}$ requests is $\thresh{e} \cdot \ell_e \geq 2 \cdot \alpha$,
the total cost of including $e$ in the matching is dominated by the half of the cost 
of serving requests to $e$. Furthermore, as the number of removals from the matching
cannot be larger than the number of additions, the total cost of excluding $e$ 
from the matching is also dominated by the same amount.
Summing up, the matching reconfiguration cost of \ALG 
is not larger than its request serving cost, i.e.,
\[
	\cost(\ALG,\I)
		\leq 4 \cdot |\VV| \cdot (\alpha + \ell_{\max}) + 2 \sum_{d \in D(\I)} 
		\ell(\assign{d}),
\]
which concludes the lemma.
\end{proof}

\begin{lemma}
\label{lem:opt_cost}
Let $D(\I)$ be the set of all desaturation events that occurred during 
input $\I$.
Then 
\[
	\cost(\OPT,\I) \geq 
	\frac{1}{3 \cdot (b+1) \cdot (1+\ell_{\max}/\alpha)} 
	\sum_{d \in D(\I)} \ell(\assign{d})  \,.
\]
\end{lemma}

\begin{proof}
To estimate the cost of \OPT, it is more convenient to think that its 
cost is not associated with node pairs but with nodes. That is, we
distribute the cost of \OPT pertaining to node pairs 
(paying for a request, including
an edge in the matching or removing an~edge from the matching) equally 
between the endpoints: When \OPT pays $\ell_\tau$ for a request
at node pair $\tau = \{u,v\}$, we account cost $\ell_\tau/2$ for node $u$ and cost 
$\ell_\tau/2$ for node $v$.
When \OPT pays $\alpha$ for including node pair $\{u,v\}$ into the matching
or excluding it from the matching, we associate cost $\alpha/2$ with node $u$ 
and $\alpha/2$ with node~$v$.

Now, fix a desaturation event $d$ at a node $w$. 
Let $d_0$ be the previous
desaturation event at node $w$. (If $d$ is the first desaturation event at~$w$,
then $d_0$ is the beginning of the input $\I$.) 
Note that all requests of $\assign{d}$
appeared between $d_0$ and $d$.

Let the node cost (in \OPT's solution) of $w$ between $d_0$ and $d$ be denoted
$\cost(\OPT,d)$. As each node-cost paid by \OPT is covered by
at most one term $\cost(\OPT,d)$, it holds that
$\cost(\OPT,\I) \geq \sum_{d \in D(\I)} \cost(\OPT,d)$. 
Hence, our goal is to lower bound the value of $\cost(\OPT,d)$
for any desaturation event $d$ (at some node $w$).

We sort the edges from $\edges{w}$ by the cost of serving them between $d_0$ and $d$.
That is, let $\edges{w} = \{e_1,e_2,\ldots,e_{|V|-1} \}$ and 
\[
	\ell(\assigne{d}{e_1}) \geq \ell(\assigne{d}{e_2}) \geq \dots 
	 \geq \ell(\assigne{d}{e_{|V|-1}})\,.
\]

Let $k$ be the number of node pairs from $\edges{w}$ that \OPT 
added to the matching between $d_0$ and $d$. The corresponding node 
cost of $w$ due to matching changes is then at
least $k \cdot \alpha/2$. Then, the total number of all node pairs 
from~$\edges{w}$ that \OPT may have in the matching at some time 
between $d_0$ and $d$ is at most $b + k$. Therefore, $\OPT$ pays 
for requests from $\assign{d}$ to all node pairs but at most $b+k$ node pairs, i.e.,
\[
	\cost(\OPT,d) \geq 
	k \cdot \alpha / 2 + \ell(\assign{d}) - \sum_{j=1}^{b+k} \ell(\assigne{d}{e_j}) \,.
\]

To lower-bound this amount, we first observe that for any $k \geq 0$, it holds that
\begin{equation}
\label{eq:opt_lower}
	3 \cdot (b+1) \cdot \left[
		\ell(\assigne{d}{e_{b+k+1}}) + k \cdot \alpha / 2
	\right] \geq (b+k+1) \cdot \alpha
\end{equation}
Indeed, if $k = 0$, then 
by Property~\ref{item:b_or_b1_with_thresh}
of Observation~\ref{obs:desaturation}), 
$\ell(\assigne{d}{e_{b+1}}) \geq \alpha/2$, and thus 
\eqref{eq:opt_lower} follows.
If $k \geq 1$,
then $(3/2) \cdot (b+1) \cdot k \geq b+k+1$ holds for any $b \geq 1$,
which implies \eqref{eq:opt_lower}.


Second, by Property~\ref{item:at_most_tresh} of Observation~\ref{obs:desaturation},
for each node pair $e \in \edges{w}$, it holds that 
\begin{equation}
\label{eq:opt_lower_2}
	\ell(\assigne{d}{e}) \leq \thresh{e} \cdot \ell_e \leq 2 \cdot (\alpha + \ell_{\max})
	= 2 \alpha \cdot (1 + \ell_{\max} / \alpha)	\,.
\end{equation}

Therefore, using \eqref{eq:opt_lower} and \eqref{eq:opt_lower_2}, we obtain that
\begin{align*}
\cost(\OPT,d)
 \geq &\; k \cdot \alpha / 2 + \sum_{j=b+k+1}^{|V|-1} \ell(\assigne{d}{e_j}) \\
= &\; k \cdot \alpha / 2 + \ell(\assigne{d}{e_{b+k+1}}) + \sum_{j=b+k+2}^{|V|-1} \ell(\assigne{d}{e_j}) \\
\geq &\; \frac{(b+k+1) \cdot \alpha}{3 \cdot (b+1)} + \sum_{j=b+k+2}^{|V|-1} \ell(\assigne{d}{e_j}) \\
\geq &\; \frac{\sum_{j=1}^{b+k+1} \ell(\assigne{d}{e_j}) }{6 \cdot (b+1) \cdot (1+\ell_{\max}/\alpha)} + \sum_{j=b+k+2}^{|V|-1} \ell(\assigne{d}{e_j}) \\
\geq &\; \frac{1}{6 \cdot (b+1) \cdot (1+\ell_{\max}/\alpha)} \cdot \sum_{j=1}^{|V|-1} \ell(\assigne{d}{e_j}) \\
= &\; \frac{\ell(\assign{d})}{6 \cdot (b+1) \cdot (1+\ell_{\max}/\alpha)} \,.
\end{align*}
Summing this relation over all desaturation events from the input~$\I$ and using 
the relation 
$$\cost(\OPT,\I) \geq \sum_{d \in D(\I)} \cost(\OPT,d)$$ 
yields the lemma.
\end{proof}

\begin{theorem}
\ALG is $O((1+\ell_{\max}/\alpha) \cdot b)$-competitive.
\end{theorem}

\begin{proof}
Fix any input instance $I$ and let $D(\I)$ be the number of desaturation 
events that occurred when \ALG was executed on $\I$. 
By Lemmas~\ref{lem:alg_cost} and~\ref{lem:opt_cost}, we immediately 
obtain that
$
\cost(\ALG,\I) 
\leq 12 \cdot (b+1) \cdot (1+\ell_{\max}/\alpha) \cdot \cost(\OPT,\I)
	+ 4 \cdot |\VV| \cdot (\alpha + \ell_{\max})
$, i.e.,
the competitive ratio is at most 
$12 \cdot (b+1) \cdot (1+\ell_{\max}/\alpha) = O((1+\ell_{\max}/\alpha) \cdot b)$.
\end{proof}


\section{Lower Bound}\label{sec:lower}
	
\begin{theorem}
The competitive ratio of any deterministic algorithm $\DET$ is at least $b$.
\end{theorem}

\begin{proof}
Let our graph be a star of $b+2$ nodes $v_0, v_1, \dots,$ \linebreak $v_b, v_{b+1}$ 
and non-reconfigurable edge set
\[
	F = \left \{\, (v_0, v_1), (v_0, v_2), \dots, (v_0, v_{b+1}) \,\right\}.
\]
Each edge of $F$ has length $1$. We start with
any matching that connects $v_0$ to $b$ leaves. At any time, the adversary chooses $v_i$ 
which is not currently matched with $v_0$, and requests a node pair $(v_0,v_i)$ 
for $\alpha$ times. These $\alpha$ requests constitute one chunk.

For each chunk, \DET pays at least $\alpha$: 
either for modifying the matching or for bypassing all $\alpha$ requests.
An offline algorithm \OFF (that knows the entire input sequence) could
however make a smarter selection of an edge to remove from the matching: \OFF 
chooses the one which is not
going to be requested in the nearest $b$ rounds. Hence, \OFF pays at most
$\alpha \cdot \lceil k / b \rceil$ for $k$ chunks of the input.
For growing $k$, the ratio between the costs of \DET and \OFF 
becomes arbitrarily close to $b$, and hence the lemma follows. 
\end{proof}


\section{Simulations}\label{sec:simulations}

In order to complement our theoretical contribution
and analytical results on the competitive ratio
in the worst case, we conducted extensive simulations,
evaluating our algorithms on real-world traffic
traces. 
In the following, we report on our main
results. 

\subsection{Methodology}\label{sec:methodology}

All our algorithms are implemented
in Python (3.7.3), using 
the graph library NetworkX (2.3.2). 
All simulations were conducted on a machine with two Intel Xeon E5-2697V3 processors with 2.6 GHz, 128 GB RAM, and 14 cores each. 

Our simulations are based on the following workloads:
\begin{itemize}
	\item \emph{Facebook~\cite{roy2015inside}:}
	We use the batch processing trace (Had\-oop)
	from one of Facebook datacenters, as well as traces
	from one of Facebook's database clusters.
	
	\item \emph{Microsoft~\cite{projector}:}
	This data set is simply a probability
	distribution, describing the rack-to-rack
	communication (a traffic matrix).
	In order to generate a trace, we sample
	from this distribution \emph{i.i.d.}
	Hence, this trace does not
	contain any temporal structure by design
	(e.g., is not bursty)~\cite{sigmetrics20complexity}.
	However, it is known that it contains
	significant spatial structure (i.e.,
	is skewed). 
	
	\item \emph{pFabric~\cite{pfabric}:}
	This is a synthetic trace and we run
	the NS2 simulation script obtained from the authors of the paper
	to generate a trace. 
\end{itemize}

In order to evaluate our algorithm, we are comparing 
four different scenarios in our simulations:
\begin{itemize}
	\item \emph{Oblivious:} The network topology is fixed and not
	optimized towards the workload by adding reconfigurable links. 
		\item \emph{Static:} The network topology is enhanced with
		an optimal static $b$-matching, computed with the perfect
		knowledge of the workload ahead of time. 
		\item \emph{Online BMA:} The online algorithm described in this paper.
		\item \emph{LRU BMA:} Like online BMA, however, 
		the cache is now managed according to a 
		least-recently used (LRU) strategy. In other words, when a link
		needs to be cached and the cache is full, the least recently used
		link in the cache is evicted.
\end{itemize}


\begin{figure*}
\begin{center}
	   \includegraphics[width=.3\textwidth]{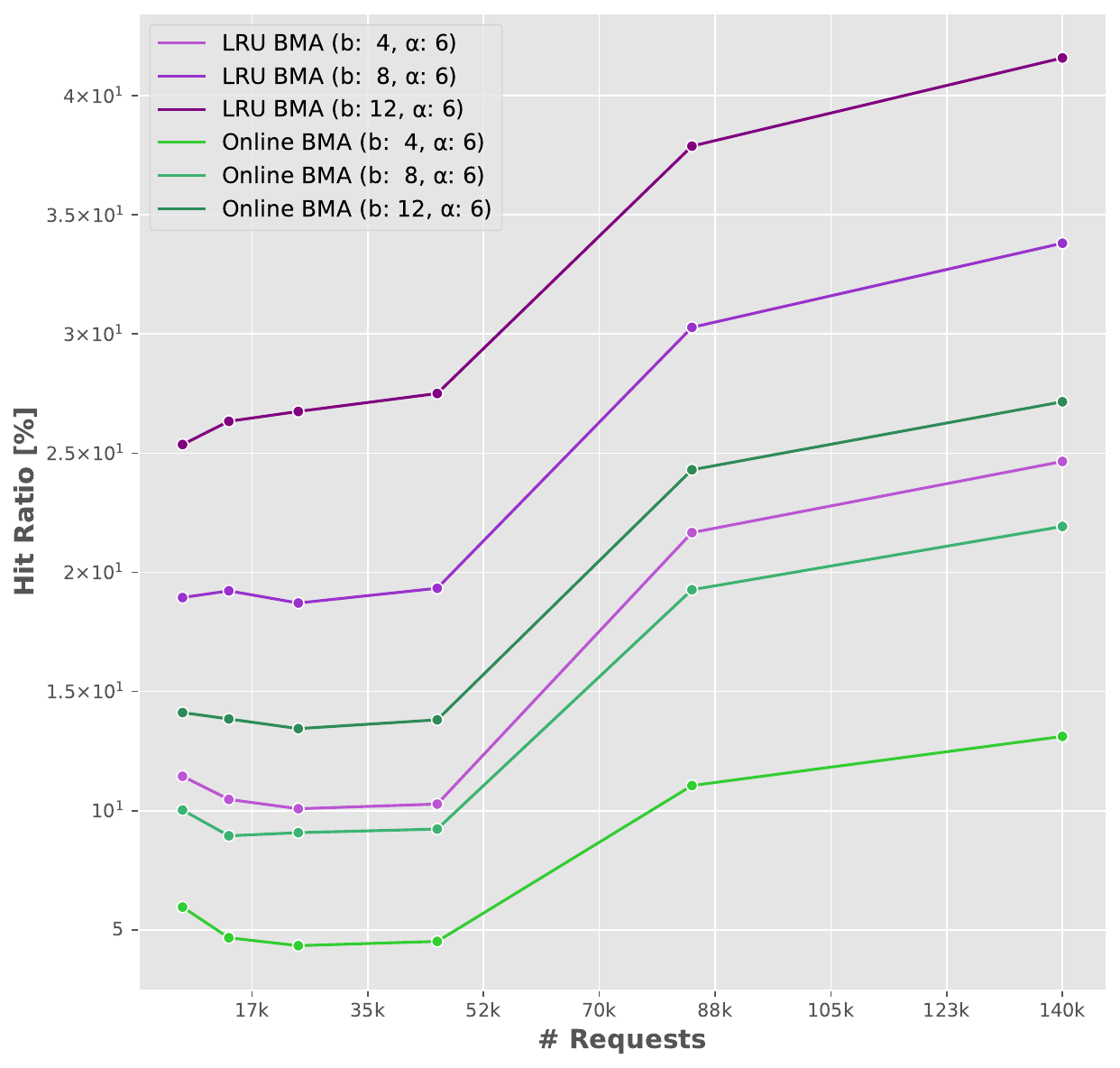}
	   \includegraphics[width=.3\textwidth]{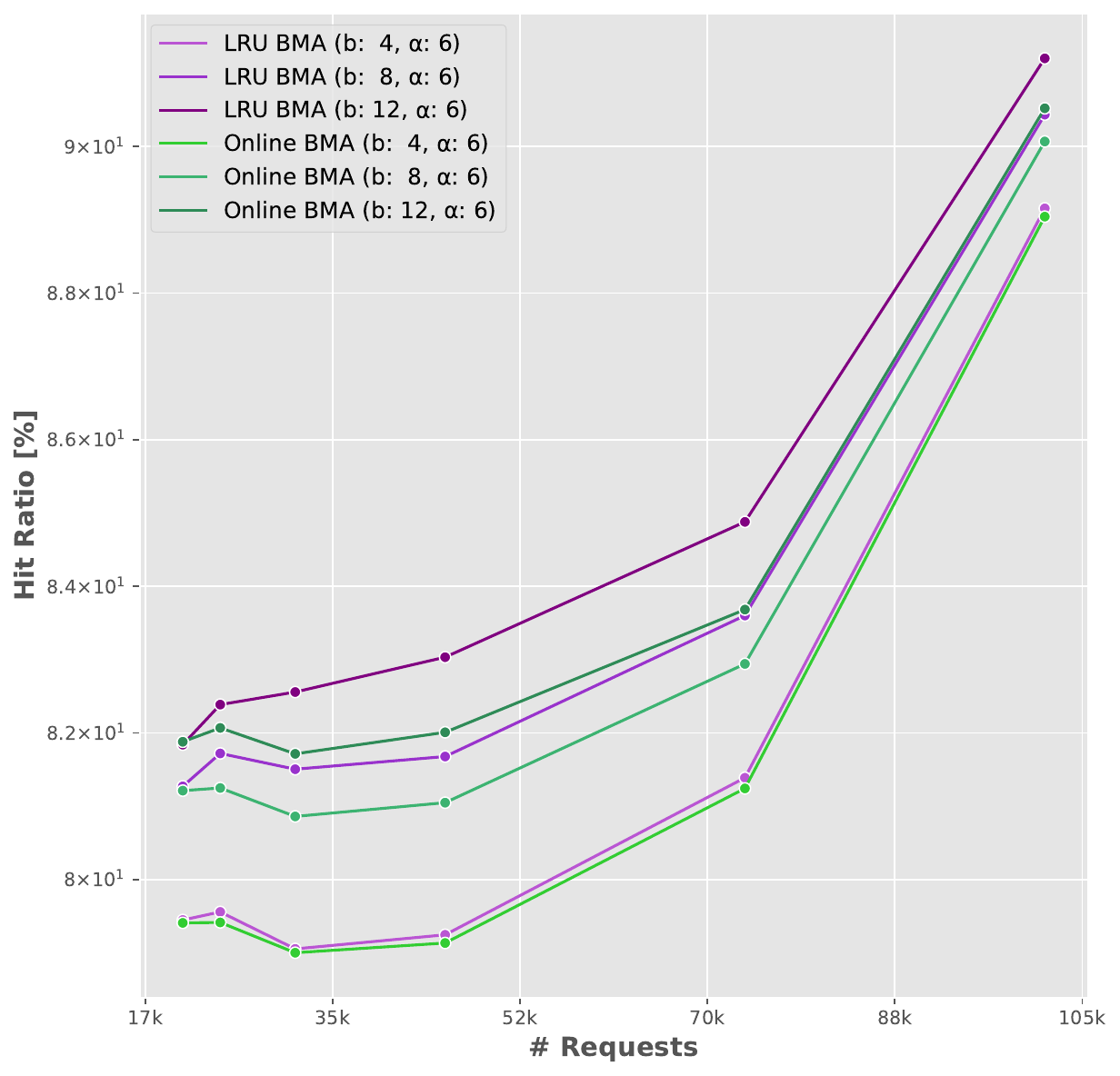}
	   \includegraphics[width=0.3\textwidth]{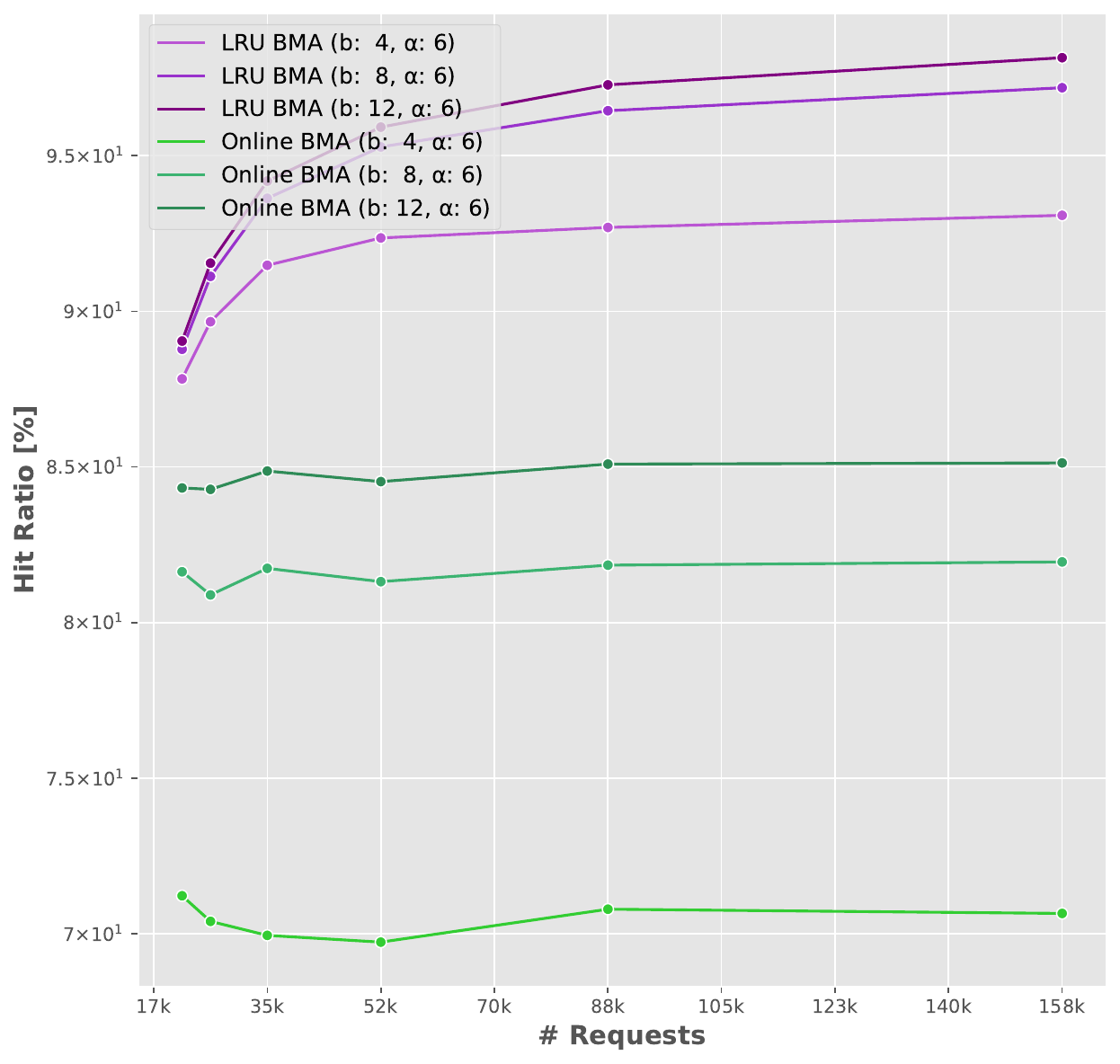}
		\\
	\caption{
	\emph{Left: Hit ratio for Facebook database cluster trace (with lower temporal locality),
	Middle: pFabric trace (with high temporal locality),
	Right: Microsoft trace (with high spatial locality).}
	}
	\label{fig:HitRatioEval}
\end{center}
\end{figure*}

\begin{figure*}[t]
\begin{center}
	   \includegraphics[width=.30\textwidth]{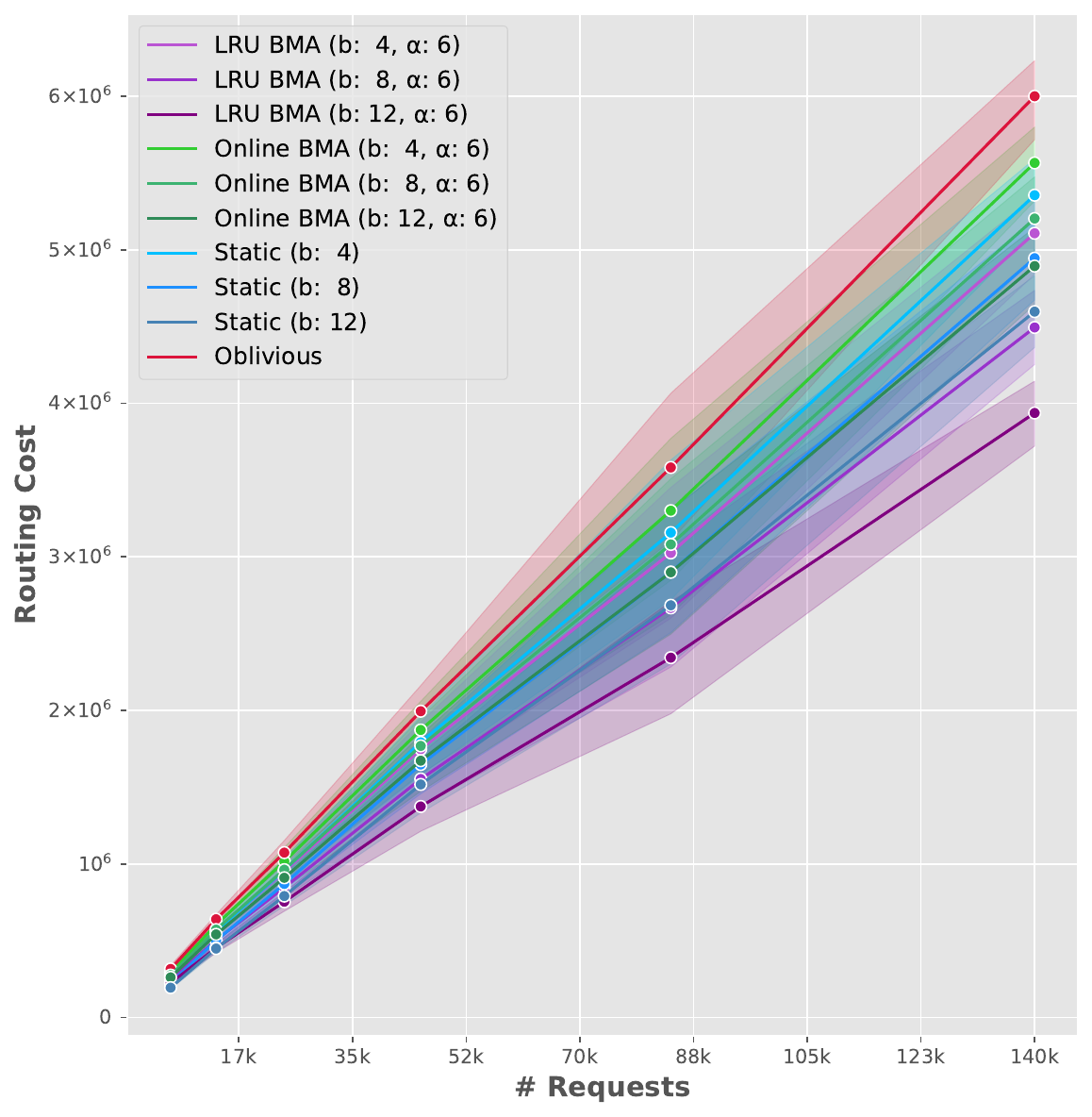}
	   \includegraphics[width=.32\textwidth]{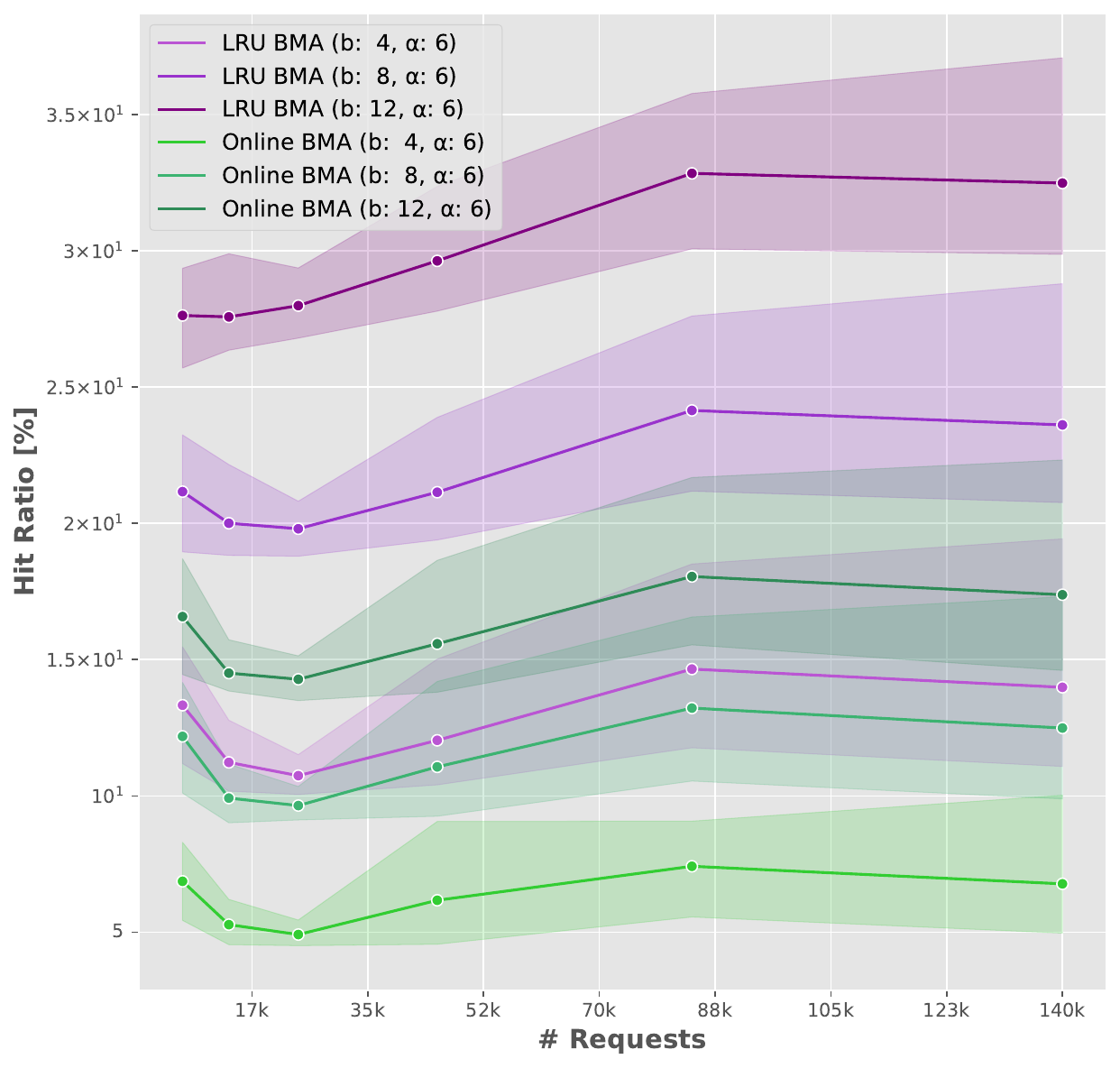} 
	   \includegraphics[width=.30\textwidth]{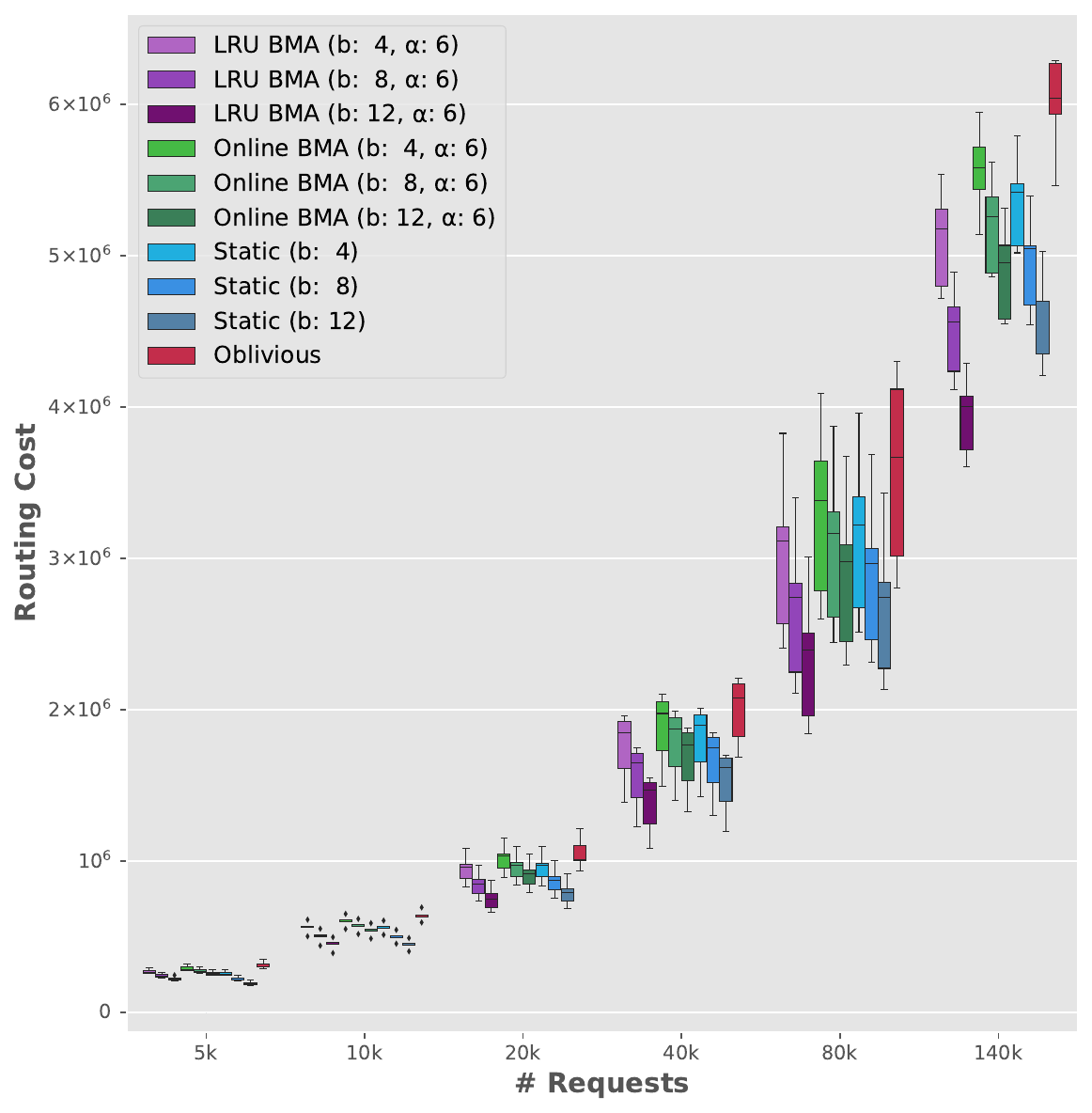}\\
	\caption{
	\emph{Left and Right: Facebook Hadoop cluster: routing costs.
	Middle: Facebook Hadoop cluster: hit ratio for different cache sizes.}
	}
	\label{fig:cacheSizeEval}
\end{center}
\end{figure*}

For all simulations, we assume a 
Clos-like datacenter topology~\cite{alFares2008},
connecting 100 servers (leaf nodes of the Clos topology).
In addition, the number of requests for each of our simulations 
depends on the actual trace, 
therefore the simulations on the Facebook cluster 
have a slightly different amount of requests 
than e.g., the Microsoft trace data.
Each test run was performed with six different request counts.
The simulations were repeated 5 times, each time with a different subset of the whole dataset to account for certain variance in the data;
the presented results are averaged over these simulation runs.
We evaluated our algorithms with several values for $b \in \{4, 8, 12\}$ and $\alpha=6$.
Note that for larger $b$, less traffic will be routed over the static network, given 
our cost function. Given this, and the fact that reconfigurable links require space,
we will be particularly interested in relatively small values of $b$: 
only a small fraction of all possible  $n\cdot (n-1)$ links is actually used.
Evaluating the effectiveness of small values for $b$ is hence
not only more interesting, but also more realistic.


\begin{figure*}[t]
\begin{center}
	   \includegraphics[width=.4\textwidth]{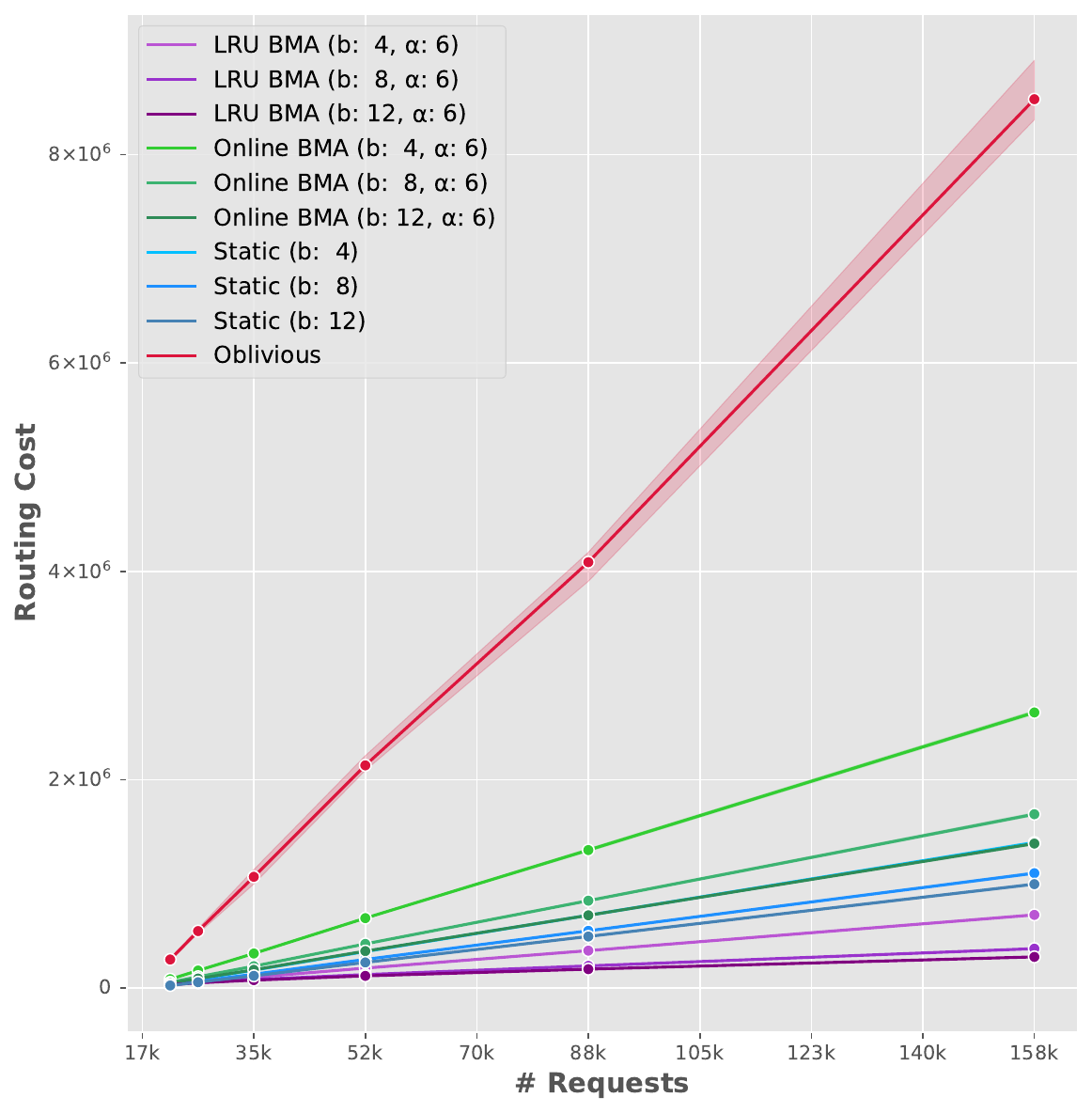}
		~
	   \includegraphics[width=.42\textwidth]{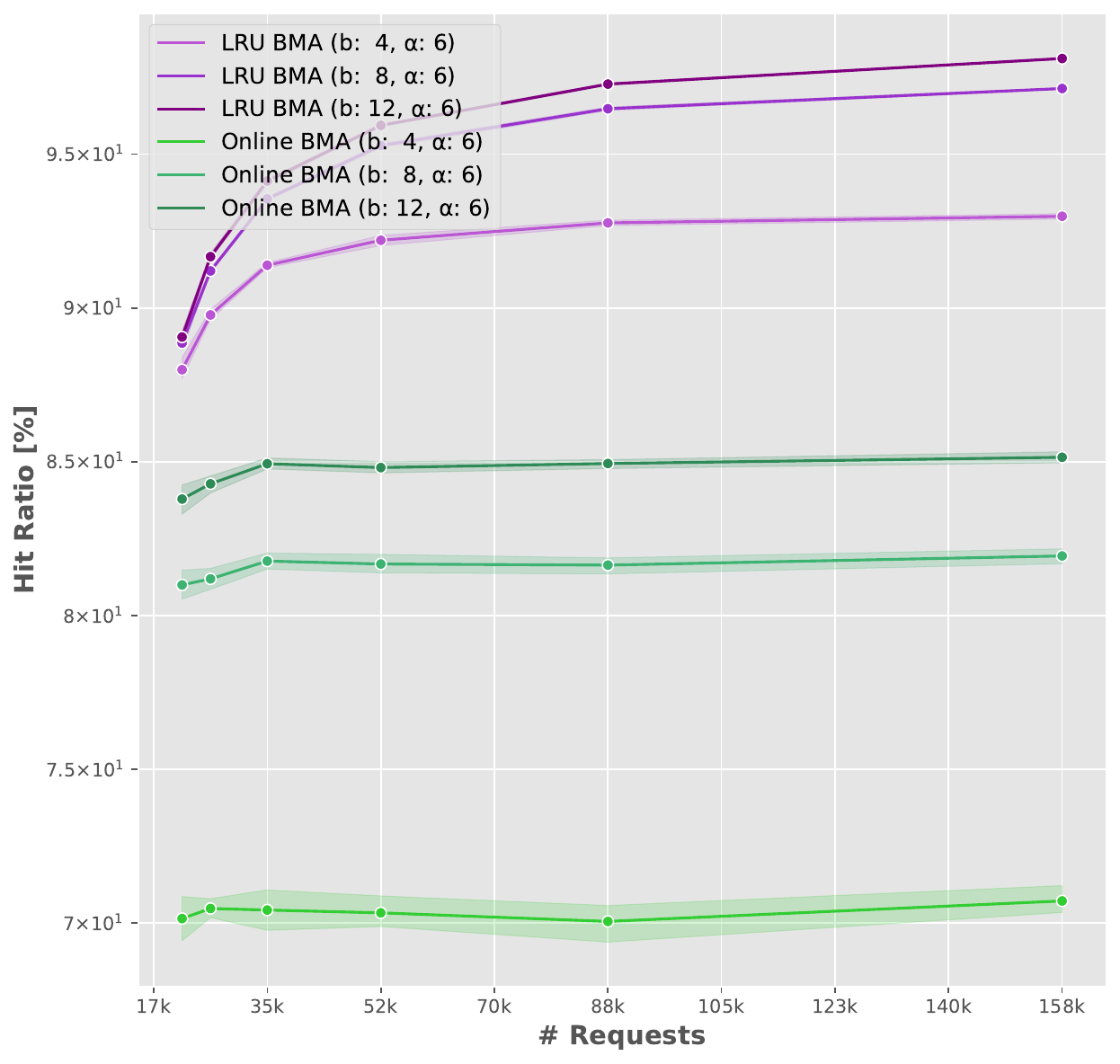}
		\\
	\caption{
	\emph{Left:} Microsoft ProjecToR: routing costs.
	\emph{Right:} Microsoft ProjecToR: hit ratio.
	}
	\label{fig:MicrosoftEval}
\end{center}
\end{figure*}

\subsection{Results}

In order to study to which extent the
\emph{Online BMA} algorithm 
can leverage the temporal locality available
in traffic traces,
we first consider the effectiveness of the 
link cache, as a microbenchmark. 
Figure~\ref{fig:HitRatioEval} shows a comparison of the hit ratio of Facebook's database traces (left), pFabric traces (right) and Microsoft traces. 
We can observe that in the case of the pFabric and Microsoft traces,
a relatively high hit ratio is obtained after a short warm-up period,
especially if a least-recently-used (LRU BMA) caching strategy is used.
We can also observe that our online algorithm performs
better under the pFabric and Microsoft traces, 
which is expected: empirical studies have 
already shown that these traces 
feature more structure than the batch processing 
traces~\cite{sigmetrics20complexity}.
We also find that the results naturally depend 
on the cache size, see Figure~\ref{fig:cacheSizeEval} (left).
An important remark is, that the degree
$b$ need to be understood relative to the total number of
switch ports, i.e., similar results are obtained for
relatively larger $b$ values.

It is interesting to compare the results of our online
algorithms to demand-oblivious topologies
as well as to static topologies.
Figure~\ref{fig:cacheSizeEval} gives a comprehensive overview of our algorithms performance in terms of route lengths (left and right plot) and also regarding the cache hit ratio (middle plot) for different cache sizes for the Facebook Hadoop cluster.
Notably, Figure~\ref{fig:cacheSizeEval} (left and middle plot) gives 
insights into our algorithm's performance over all 5 test runs, illustrating the average result, as well as the maximum and minimum result (shaded areas).   

As expected, \emph{Oblivious} always performs worse than \emph{Static}, 
\emph{Online BMA} and \emph{LRU BMA}. 
We further observe that the performance of \emph{Online BMA} 
comes close to the performance of \emph{Static},
which knows the demands \emph{ahead of time} (but is fixed).
We expect that under longer request sequences, 
when larger shifts in the communication patterns
are likely to appear, the online approach
will outperform the static offline algorithm.
To investigate this, however, the publicly available traffic traces
are not sufficient.

While the Microsoft trace does not contain temporal
structure as it is sampled i.i.d., 
it can still be exploited toward a more efficient routing
and yield a very high cache hit ratio,
due to its spatial structure, i.e., the skewed
traffic matrix. See Figure~ \ref{fig:MicrosoftEval}.

In conclusion, while our main contribution in this paper
concerns the theoretical result, we observe that our 
online algorithm performs fairly well under real-world
workloads, even without further optimizations (besides
an improved cache eviction strategy).


\section{Related Work}\label{sec:relwork}

Reconfigurable networks based on
optical circuit switches, 60 GHz 
wireless, and free-space optics,  
have received much attention over the last years~\cite{avin2019renets,firefly,helios,projector,reactor,teh2020flexspander,zhou2012mirror},
see also the recent survey for an overview~\cite{sigact19}.
For an overview of the theoretical foundations of demand-aware
networks, we refer the reader to~\cite{ccr18san}.
It has been shown empirically that reconfigurable networks
can achieve a performance similar to a demand-oblivious
full-bisection bandwidth network at significantly lower cost~\cite{firefly,projector}. 
Furthermore, the study of reconfigurable networks is not limited
to datacenters and interesting use cases also arise
in the context of wide-area networks~\cite{Jia2017,jin2016optimizing}
and overlays~\cite{ratnasamy2002topologically,scheideler2009distributed}.

Our paper is primarily concerned with the \emph{algorithmic}
problems introduced by such technologies.
The basic problem of how to optimally enhance 
a given fixed datacenter topology with reconfigurable links
has already been studied in the literature, both in the static setting, 
e.g..,~\cite{fenz2019efficient,ancs18,ccr19danr},
as well as in the dynamic setting~\cite{avin2019renets,dinitz2020scheduling,eclipse}. 
For example, Avin et al.~\cite{avin2019renets} presented a statically optimal dynamic network, ReNets,
which finds an optimal tradeoff between the benefits and costs of reconfigurations.
The static problem is related to \emph{graph augmentation}
models, which consider the problem of adding edges to a given
graph, so that path lengths are reduced. For example, 
Meyerson and Tagiku~\cite{meyerson2009minimizing} study how to add ``shortcut edges''
to minimize the average shortest path distances, 
Bil{\`o} et al.~\cite{bilo2012improved} and
Demaine and Zadimoghaddam~\cite{demaine2010minimizing}
study how to augment a network to reduce its diameter,
and there are several interesting results on how to add ``ghost edges'' 
to a graph such that it becomes (more) ``small world''~\cite{gozzard2018converting,ghost-edges,small-world-shortcut}.
However, these edge additions can be optimized globally and in a biased manner, and
hence do not form a matching. In particular, it is impractical (and does not scale) to add many
flexible links per node in practice. 
The dynamic setting is related to classic switch scheduling
problems~\cite{chuang1999matching,mckeown1999islip}.

Another line of related works considers the design of demand-aware networks
from scratch~\cite{rdan,dan,infocom19dan,Huq2017LocallySS}, ignoring the fixed topology which is available in current 
architectures (and in the near future).
Also for this model, dynamic approaches which come with provable guarantees,
have been proposed in the literature, e.g., SplayNets~\cite{ton15splay}
and Push-Down Trees~\cite{latin20}. 

In this paper, we initiated the study of an online
version of the dynamic $b$-matching problem. A polynomial-time algorithm for the static 
version of this problem has already been presented 
over 30 years ago~\cite{anstee1987polynomial,Schrij03},
and the problem still receives attention today due
to its numerous applications, for example
in settings where customers in a market need to be matched to a 
cardinality-constrained set of items (e.g., 
matching children to schools, reviewers to papers, or
donor organs to patients), but also in protein structure alignment, 
computer vision, estimating text similarity, etc.

Note that there is a line of papers studying 
(bipartite) online matching
variants
\cite{online-matching-simple,adwords-primal-dual,concave-matching,ranking-primal-dual,online-matching,bipartite-matching-strongly-lp,adwords-lp,adwords-ec}.
This problem attracted significant attention in the last decade because of its connection
to online auctions and the AdWords problem~\cite{adwords-survey}. 
Despite similarity in names (e.g., the
bipartite (static) $b$-matching variant was considered
in~\cite{kalyanasundaram2000optimal}), this model is fundamentally
different from ours. That is, it considers bipartite graphs in which nodes and
(weighted) edges appear in time and the algorithm has to choose a subset of
edges being a matching. In our scenario, the (non-bipartite) graph is given a
priori, and the algorithm has to \emph{maintain a dynamic} matching. One way of
looking at our scenario is to consider the case where edges weights can change
over time and the matching maintained by an algorithm needs to catch up with such
changes.


\section{Conclusion}\label{sec:conclusion}

Motivated by emerging reconfigurable datacenter networks
whose topology can be dynamically optimized toward the workload, 
we initiated the study of a fundamental problem, online $b$-matching.
In particular, we presented competitive online algorithms 
which find an optimal trade-off between the benefits and costs of
reconfiguring the matching.
While our main contribution concerns the derived theoretical 
results (i.e., 
the competitive online algorithm and the lower bound),
we believe that our approach has several interesting practical
implications: our algorithm is simple to implement, has a low runtime
and, as we have shown, performs fairly well also under different real-world
workloads and synthetic traffic traces. 

Our work opens several interesting avenues for future research.
In particular, we have so far focused on deterministic algorithms,
and it would be interesting to explore randomized approaches; 
in fact, our first investigations in this direction indicate that 
the randomized setting is more challenging, also due to
the introduced dependencies, and exploiting the connection to online
paging is difficult.
On the practical side, it 
would be interesting to investigate specific reconfigurable optical technologies as well as 
specific datacenter topologies (such as Clos topologies) in more details, 
and tailor our algorithms and develop distributed implementations for 
an optimal performance in this case study.

{
\balance
\bibliographystyle{abbrv}
\bibliography{refs}
}

\end{document}